\documentclass[10pt,journal,compsoc]{IEEEtran}
\usepackage{amsmath}
\usepackage{amssymb, amsthm, amsfonts, mathrsfs,dsfont}
\usepackage{graphicx}         
\usepackage[dvips]{epsfig}    
\usepackage{xcolor}
\usepackage{enumitem}
\usepackage{algorithm}
\usepackage{algorithmic} 
\usepackage{tikz}
\usetikzlibrary{arrows.meta,fit,shapes,snakes}

\newtheorem{coro}{Corollary} 
\newtheorem{rem}{Remark} 
\newtheorem{defn}{Definition}

\newtheorem{thm}{Theorem}
\newtheorem{prop}{Proposition}
\newtheorem{exmp}{Example}
\newtheorem{prob}{Problem}

\newcommand{\subscr}[2]{#1_{\textup{#2}}}

\title{Functional target controllability of networks: structural properties and efficient algorithms}% minimal steering set}
\author{Christian Commault\thanks{Univ. Grenoble Alpes, CNRS, GIPSA-lab, F-38000 Grenoble, France. Email:  christian.commault@gipsa-lab.grenoble-inp.fr} \quad\and\quad 
Jacob van der Woude\thanks{DIAM, EWI, Delft University of Technology, Delft, the Netherlands,  Email: J.W.vanderWoude@tudelft.nl}
\quad\and\quad Paolo Frasca\thanks{Univ. Grenoble Alpes, CNRS, Inria, GIPSA-lab, F-38000 Grenoble, France. paolo.frasca@gipsa-lab.grenoble-inp.fr}}
                     
\begin{document}
%\begin{frontmatter}
%\address[label1]{Univ. Grenoble Alpes, CNRS, GIPSA-lab, F-38000 Grenoble, France. Email:\{christian.commault,paolo.frasca\}@gipsa-lab.grenoble-inp.fr}
%\address[label2]{DIAM, EWI, Delft University of Technology\\van Mourik Broekmanweg 6, 2628 XE, Delft, the Netherlands,  Email:J.W.vanderWoude@tudelft.nl}

\maketitle

\begin{abstract}
In this paper we consider the problem of controlling a limited number of target nodes of a network.  
Equivalently, we can see this problem as controlling the target variables of a structured system, where the state variables of the system are associated to the nodes of the network.
We deal with this problem from a different point of view as compared to most recent literature. 
Indeed, instead of considering controllability in the Kalman sense, that is, as the ability to drive the target states to a desired value, we consider the stronger requirement of driving the target variables as time functions. The latter notion is called functional target controllability.
We think that restricting the controllability requirement to a limited set of important variables justifies using a more accurate notion of controllability for these variables. 
Remarkably, the notion of functional controllability allows formulating very simple graphical conditions for target controllability in the spirit of the structural approach to controllability. 
The functional approach enables us, moreover, to determine the smallest set of steering nodes that need to be actuated to ensure target controllability, where these steering nodes are constrained to belong to a given set. We show that such a smallest set can be found in polynomial time. % that can actually be actuated. 
We are also able to classify the possible actuated variables in terms of their importance with respect to the functional target controllability problem.
\end{abstract}
%\begin{keyword}
%Controllability, Structured system theory, Input-output linkings
%\end{keyword}
%\end{frontmatter}

\section{Introduction}
The network community has shown a real interest in control concepts in the recent years \cite{STRO:01,LOM:07,SOR:07} and the control community has reciprocated by a growing interest in network applications% of tools from systems theory
~\cite{ME-SM-MC-CK-AB:12,GP-GN:12,FP-SZ-FB:14,ols:14,GB-AG-PF-FP:15,YZ-JC:17,GL-CA:18}. 
Most of the papers in literature study controllability according to the most common definition in systems theory, that is, the ability to steer the state to a target point: we shall refer to this definition as {\em point-wise} controllability. This definition (and the criticism of its limits~\cite{Cow:12,FP-SZ-FB:14}) have also been the starting point of a series of works that approached more advanced questions like quantifying the energy required for control \cite{YAN:12,Kli:17,GL-CA:18} and the robustness of the controllability properties~\cite{TM-VK-DB-FP:18,GB-AG-PF-FP:15,SCJ-MW-MZ-RAC:18,FP-CF-SZ-SZ:18} in the context of networks.

For linear systems, the classical notion of point-wise controllability lends itself to what is called the ``structural'' approach, as it was introduced four decades ago by Lin \cite{Lin:74,MU:87}. 
In this approach, the controllability properties are characterised in terms of the sole topology of the network associated to the system: the potential of these methods in network science has become manifest in the last few years~\cite{Liu:11,Ruths:14,TM-DB-FP:18}.  
The controllability problem under study in these references is known as the Minimum Input Problem \cite{Liu:11}, or more generally the Minimum Controllability Problem \cite{ols:14}, and can be formulated as follows.
Given an autonomous dynamical system, one looks for the minimal number of driver nodes (nodes which are directly connected with a control input) such that we have a point-wise control of all the states, {\em i.e.,} we are able to drive the global state from any initial point to any final point in any fixed positive time.

Since controlling the whole state may be too demanding, and often not necessary, several authors have dealt with the so-called ``target control'' \cite{Gao:14,Kli:17,LI:18,WAA:17}. 
In this case, one defines a set of important variables and requires to control the corresponding states only: this choice of course induces a relaxation of the controllability conditions, since the system no longer needs  be controllable in the usual, full-state, sense. Furthermore, target controllability can in principle be checked by using a natural extension of the Kalman controllability condition. 
However, this definition is not so easy to deploy as it seems, for two main reasons. 

The first difficulty is the inability to leverage the structural approach. A structural characterisation of target controllability was left as an open problem in \cite{MU:90} and, to the best of our knowledge, no graph characterisation for structural target controllability is available to date (unlike for structural point-wise state controllability). In the literature on structural target controllability, the authors either develop approximate approaches \cite{Gao:14,WU:15}, or study particular types of systems: for instance, the problem has an elegant solution when the dynamics of the system is symmetric \cite{LI:18}.
 
The second difficulty is the intrinsic hardness of the problem: the Minimum Controllability Problem for target controllability has recently been proved to be a NP-hard problem \cite{CZE:18}. This negative result implies the need for heuristic solutions in practical situations: one such approximate algorithm is provided in \cite{CZE:18}.

%However,  trying to solve the target controllability problem in the context of structured systems raises several difficulties.
%\begin{itemize}
%\item The structural target controllability was left as an open problem in \cite{MU:90}, and to the best of our knowledge, unlike for structural point-wise state control, no graph characterisation for structural target controllability is available today. 
%In the literature on structural target controllability, the authors either study particular types of systems, or develop approximate approaches \cite{Gao:14,WU:15}. 
%Interestingly, the problem has an elegant solution when the dynamics of the system is symmetric  \cite{LI:18}.
%\item Solving the Minimum Controllability Problem for target controllability has recently been proved to be a NP-hard problem \cite{CZE:18}.
%\end{itemize}
To overcome these difficulties, we propose in this paper to use a different point of view on structural target controllability. 
We consider here  {\em functional output controllability}, {\em i.e.,} we ask for the possibility to follow any output profile, and not only the possibility to reach any particular point in the output space. 
This is more demanding than the usual (point-wise) controllability.
In particular, when the number of inputs is less than the number of states, the whole state space cannot be functionally controllable. Moreover, we assume that the steering nodes must be chosen in a given set, defined by physical or technological considerations, that we call available nodes \cite{OLS:15}.

Our opinion is that, since we only concentrate on some important variables, a more accurate controllability can be desirable and be afforded for these variables.
To illustrate our point of view, we may think of examples from as diverse domains as automotive and drug delivery. 
Let us think about a car with an automatic gear box. 
The main control actuators are the accelerator, the driving wheel and the brake. 
With these three controls, the driver can influence the hundreds of variables which can be listed in a reasonable physical model of the vehicle. 
However, for the driver, very few of these variables are really important, roughly speaking, only the velocity and the direction of the vehicle are essential. On another hand, for these two variables, we need a precise control of their time behaviour, not only a point-wise control.
In biology networks, in particular for pathology treatment \cite{WU:15,KA:17}, only some variables are essential to be controlled, and it is of interest to identify the nodes (or cells) where drugs must be applied to avoid the abnormal behaviour of these essential variables. 
For these sensitive health parameters, it is thought that a trajectory controllability is much more desirable than the possibility to evolve from one value to another one without mastering the transient behaviour, which could be dangerous. 
Moreover, is clear that for this type of application the drug cannot in general be applied to any node of the network, the steering nodes can only be chosen in a restricted set of admissible nodes.

Owing to this new point of view, the main contributions of our work can be summarized in the following four points. 
\begin{itemize}
\item  We define the notion of {\em functional target controllability} and characterise it in graph terms for structured systems. 
This characterisation (given in Corollary~\ref{OFC}) is a direct consequence of results on the structural rank of transfer matrices which appeared around three decades ago \cite{OHTA:85, CDP:91,VDW1:91,VDW2:91}.
\item We define the {\em Minimum Target Controllability Problem (MTCP)} when controllability is understood in the functional sense and with the constraint that the steering nodes have to be chosen within a given set of available nodes. By exploiting the above characterisation of controllability, we give a full solution to it (Proposition~\ref{Exist}).
Indeed, we characterise the minimum number of steering nodes and describe a procedure to find a set of steering nodes with minimum size, which ensure functional target controllability.
\item We establish (Theorem~\ref{Class2}) a {\em classification of the available nodes} depending on their importance for the functional target controllability, namely dividing them into essential, useful and useless nodes.
\item We show that solving the MTCP and classifying the available nodes can be done with {\em polynomial complexity} by using standard algorithms that solve Maximum Flow problems, such as the Ford-Fulkerson algorithm (Proposition~\ref{Sep-Cut}).
\end{itemize}

In order to present these contributions, the outline of this paper will be the following.
In Section~\ref{ProbFormul} we present the target controllability problem with the two different points of view.
In Section~\ref{LSSOBS} we recall the main results of structured systems concerning the graph characterisations of the types of controllability and we illustrate the results by two examples.
We state the main result on the Minimum Target Controllability Problem in Section~\ref{MTCP}. 
In Section~\ref{Class1} we give a classification of available states with respect to the MTCP.
Section~\ref{Algo} deals with algorithmic and complexity aspects of the main result. 
Finally, in section~\ref{Conc} we conclude the paper with some remarks and topics for future research.
\section{Output controllability: Point-wise {\em vs} functional, and problem formulation} \label{ProbFormul}
%\section{Problem formulation} \label{ProbFormul}
In this paper, we consider a large scale system composed of $n$ agents interacting together with linear dynamics. We can then represent the whole behaviour of the system by the simple equation
\begin{eqnarray}  \label{lin1}
\begin{array}{c}
\dot{x}(t)=Ax(t),
\end{array}
\end{eqnarray}
where  $x(t) \in {\mathbb R}^n$ is the state vector and $A$ is a real $n \times n$ matrix.
We will also consider the system when the dynamics is influenced by external input signals, and when some variables called outputs, which are linear combinations of state variables, give an external view of the system. 
The global system can then be represented as
\begin{eqnarray}  \label{lin2}
\Sigma:
\begin{array}{l}
\dot{x}(t)=Ax(t)+Bu(t),\\
y(t)=Cx(t),
\end{array}
\end{eqnarray}
where  $u(t) \in {\mathbb R}^m$ is the input vector, $y(t) \in {\mathbb R}^p$ is the output vector, and $B$ and $C$ are real matrices of suitable dimensions. \\
Occasionally, we will distinguish $m$ states, called the {\em steering states} $\mathcal{S}=\{x_{i_{1}}, \dots,x_{i_{m}}\}$, with $i_k\in \{1, \dots,n\}$ and $i_1<i_2< \dots<i_m$. 
To each steering state $x_{i_k}(t)$ we associate a control input $u_k(t)$ which acts only on this state variable. 
In this case, the input matrix will be denoted by $B_{\mathcal{S}}$. 
The $B_{\mathcal{S}}$ matrix has $m$ columns, and column $k$ has all its entries equal to $0$ except for $b_{i_{k}k}$. \\
Similarly, a certain number of state variables ${\mathcal T}=\{x_{j_{1}}, \dots,x_{j_{p}}\}$, with $j_l\in \{1, \dots,n\}$ and $j_1<j_2< \dots<j_p$, called {\em target variables}, are of a prominent importance. 
Each target state $x_{j_l}(t)$ is associated with a unique output $y_l(t)$. 
The set of target variables induces therefore the $C_{\mathcal T}$ matrix. 
The $C_{\mathcal T}$ matrix has $p$ rows, each row $l$ has all its entries equal to $0$ except for $c_{l j_l}$.

\subsection{Point-wise output controllability}
A first possibility for considering output controllability is to define it as an extension of the classical state controllability.
\smallskip
\begin{defn}[Point-wise controllability]\label{pointw}
The system (\ref{lin2}) is said to be {\em (point-wise) output controllable} if, for initial condition $x(0)=0$, any instant $T>0$, and any point $y_T$ of the output space $\mathbb{R}^p$, there exists an input function $u(t)$ such that $y(T)=y_T$.
\end{defn}
It is easy to see that this output controllability can be tested via an extension of the Kalman controllability condition, {\em i.e.,} the system is output controllable if and only if 
\begin{eqnarray}\label{Output-Cont}
\mbox{rank}(C[B,AB, \dots,A^{n-1}B])=p,
\end{eqnarray}
with $p$ being the number of outputs.

\subsection{Functional output controllability}
In this paper, we will prefer a refined notion of output controllability. Instead of looking for the possibility to reach any point in the output space, we will ask for the possibility to follow any output trajectory. The notion of functional output controllability was introduced first in \cite{BROC:65}, where it was called functional reproducibility. In the latter paper, functional reproducibility was also characterised for linear systems. This characterisation makes more precise the intuition which relates the possibility of finding an input producing a given output to some form of invertibility of the system. Several papers, see for example \cite{SILV:69,SAIN:69}, brought additional contributions in this area and discussed the relations between point-wise and functional output controllability.

In order to define it rigorously, we recall that a function $f(t)$ is said to be $C^\infty [t_1,t_2]$ if it is differentiable on the interval $[t_1,t_2]$ for any order of differentiation.
\smallskip
\begin{defn}[Functional controllability]\label{funct}
The system (\ref{lin2}) is said to be {\em functional output controllable} if, for initial condition $x(0)=0$, any instant $T>0$, and any $C^\infty[0,T]$ trajectory $\tilde y(t)$ in the output space $\mathbb{R}^p$, there exists an input function $u(t) \in C^\infty[0,T]$, such that the output of (\ref{lin2}) satisfies $y(t)=\tilde y(t)$ for all $t \in [0,T]$.
\end{defn}
From classical linear system theory, we recall that the transfer matrix of the system (\ref{lin2}) is the rational matrix $T(s)=C(sI_n-A)^{-1}B$, where $I_n$ is the identity matrix of size $n$. 
If we denote by $\bar u(s)$ and $\bar y(s)$, the Laplace transforms of the vector time functions $u(t)$ and $y(t)$, respectively, we have that $\bar y(s)=T(s)\bar u(s)$, when assuming zero initial conditions. 
A transfer matrix is a matrix over the field of rational functions and as such all classical properties (rank, invertibility,...)  of real matrices are applicable to matrix transfer functions. 
For a system of type (\ref{lin2}), the rank of the transfer matrix $T(s)$, sometimes called its normal rank, is defined as the rank of the matrix for almost any value of the variable $s$, with the finite number of singularities coming from the poles and the zeros of $T(s)$ \cite{KAI:80}.\\
From basic results on control, the functional output controllability can be characterised as follows.
\begin{prop}[Output controllability and transfer matrix~\cite{BROC:65}]\label{caract}
The system (\ref{lin2}) is {functional output controllable} if and only if the transfer matrix $T(s)$ has rank $p$ (the number of outputs). 
\end{prop}
The systems whose transfer matrix is such that rank $T(s)=p$ are called right invertible systems. 
Proposition~\ref{caract} means that, for right invertible systems, given an objective output function $\tilde y(t)$ (or its Laplace Transform), it is possible to find an input function $\tilde u(t)$ which can produce this output. 
This assertion assumes some smoothness of the function $\tilde y(t)$, when the input needs to remain in some physically feasible function class.
This is why we restrict here the output trajectories to be $C^\infty$. 
This notion of functional controllability is more powerful and implies the classical (point-wise) controllability, but it is of course more demanding in terms of conditions on the system. A discussion on these two points of view on output controllability appeared in \cite{HAU:83} in the context of non-interacting control. 
%Moreover, for physical or technological reasons, there may exist some nodes which cannot be directly controlled, in such a way that the steering nodes must belong to a restricted set, called the {\em available node set}, ${\mathcal A \subset \mathcal X}$. 
%We will be looking in the following for a minimum set of steering nodes from ${\mathcal A}$ such that the corresponding system defined in (\ref{lin2}) is functional output controllable.\\
\subsection{Formulation of the output controllability problem}
We are given a dynamic system, as in (\ref{lin1}), representing the network, so that the matrix $A$ is given. 
The designer has decided that a certain number of state variables ${\mathcal T} \subset X$, called target variables, are of importance. 
The set of target variables induces in one-to-one correspondence an output set $Y_{\mathcal T}$ and therefore the $C_{\mathcal T}$ matrix. 
We have now to choose a minimum number of steering nodes ${\mathcal S}$, which will define  an input set $U_{S}$, and therefore the $B_{\mathcal S}$ matrix, such that the system (\ref{lin2}) is functionally output controllable. Moreover, for physical or technological reasons, there may exist some variables which cannot be directly controlled, in such a way that the steering variables must belong to a restricted set, called the {\em available set}, ${\mathcal A \subset \mathcal X}$. 
\begin{prob}\label{Prob}
Given a system of type (\ref{lin1}) with a set ${\mathcal T}$ of target variables, characterise, within the set of available variables $\mathcal{A}$, the sets of steering variables $\mathcal{S}$ such that the system $(A,B_{\mathcal{S}},C_{\mathcal T})$ is functional output controllable. This characterisation includes:
\begin{itemize}
\item The determination of the minimum size of an admissible set of steering variables.
\item An evaluation of the importance of each available steering variable for functional output controllability.
\end{itemize}
\end{prob}
In order to fit this problem with the network paradigm, the systems will be studied in the structured system framework which has a natural graph interpretation.

\section{Linear structured systems and functional output controllability} \label{LSSOBS}
In this section we will recall first the main notions and graph tools for structured systems. 
We will then recall the well-known result on structural controllability and finally present our main result on functional output controllability. 
The concepts and results will be illustrated and compared via two examples.
\subsection{Structured systems and structural controllability}
We consider a linear system with parametrized entries denoted by
$\Sigma_\Lambda$.
\begin{eqnarray}  \label{struct}
\Sigma_{\Lambda} :
\begin{array}{lc}
\dot{x}(t)=A_{\Lambda}x(t)+B_{\Lambda}u(t), \\
y(t)=C_{\Lambda}x(t),
\end{array}
\end{eqnarray}
where $x(t) \in \mathbb{R}^n$ is the state vector, $u(t) \in \mathbb{R}^m$ the input signal, and $y(t) \in \mathbb{R}^p$ the output signal. 
Further, $A_{\Lambda}$, $B_{\Lambda}$ and $C_{\Lambda}$ are matrices of appropriate dimensions in which the non-zero entries are each replaced by a parameter, and where all parameters are collected in a parameter vector $\Lambda \in \mathbb{R}^k$. 
The system is called a linear structured system.
Clearly, the entries of the composite matrix 
\begin{eqnarray}
J_{\Lambda} =
\left [
\begin{array}{cc}
A_{\Lambda}& B_{\Lambda} \\
C_{\Lambda}& 0
\end{array}
\right ],
\end{eqnarray}
are either fixed zeros or independent parameters (not related by algebraic equations) \cite{Lin:74,MU:87,DionCV:02}.
Occasionally, we will use the system without explicit inputs and outputs as in (\ref{lin1}): the structured system will then be
\begin{eqnarray}  \label{Alambda}
\begin{array}{c}
\dot{x}(t)=A_{\Lambda}x(t).
\end{array}
\end{eqnarray}

For linear structured systems one can study generic properties, \emph{i.e.,} properties which are true for almost all values of the $k$ parameters collected in $\Lambda$. 
More precisely, a property is said to be {\em generic} (or {\em structural}) if it is true for all values of the parameter vector $\Lambda$ outside a proper algebraic variety in the parameter space $\mathbb{R}^k$. Recall that a proper algebraic variety is the intersection of the zero set of some non-trivial polynomials with real coefficients in the  $k$ parameters of the system. A proper algebraic variety has Lebesgue measure zero.\\
A directed graph $G(\Sigma_{\Lambda})=(Z,W)$ can be associated with a structured system $\Sigma_{\Lambda}$ of type (\ref{struct}).
\begin{itemize}
\item The node set is $Z=X \cup U \cup Y$, where $X$, $U$ and $Y$ are the state node set, input node set and output node set, given by $\{x_1,x_2,\ldots,x_n\}$, $\{u_1,u_2,\ldots,u_m\}$ and $\{y_1,u_2,\ldots,y_p\}$, respectively.
\item The edge set is $W= \{(x_i,x_j)|a_{\Lambda ji}   \neq 0 \}\cup \{(u_i,x_j)|$ $b_{\Lambda ji} \neq 0 \} $ $ \cup\{(x_i,y_j)|c_{\Lambda ji} \neq 0 \}$, where $a_{\Lambda ji}$ denotes the $(j,i)$th entry of $A_\Lambda$ and $(x_i,x_j)$ denotes an edge from node $x_i$ to node $x_j$, and similarly for  $b_{\Lambda ji}$ and $(u_i,x_j)$, and $c_{\Lambda ji}$ and $(x_i,y_j)$.
\end{itemize}
In the particular case (\ref{Alambda}), the graph is denoted $G(A_\Lambda)$.

A path in $G(\Sigma_\Lambda)$ from a node $v_{0}$ to a node $v_{q}$  is a sequence of edges, $(v_{0},v_{1}),(v_{1},v_{2}),\ldots,(v_{q-1},v_{q})$, such that $v_{t} \in Z$ for $t=0,1,\ldots,q$, and $ (v_{t-1},v_{t})\in W$ for $t=1,2,\ldots,q$.  
The nodes $v_0, \dots,v_q$ are then said to be {\em covered} by the path. 
A path which does not meet the same node twice is called a {\em simple} path. 
If $v_0 \in U$ and $v_q \in X$, the path is called an input-state path. 
A path for which $v_0=v_q$ is called a circuit. 
A {\em stem} is a simple input-state path. 
A system is said to be {\em input-connected} if any state node is the end node of a stem. 
A cycle is a circuit which does not meet the same node twice, except for the initial/end node. 
Two paths are disjoint when they cover disjoint sets of nodes. 
When some stems and cycles are mutually disjoint, they constitute a {\em disjoint set of stems and cycles}.\\
The following result characterises structural controllability \cite{Lin:74,Reinschke:88,SH:76}.
\begin{thm}[Structural characterisation of point-wise controllability]\label{lin}
\label{obsdil1} Let $\Sigma_{\Lambda}$  be the linear structured system defined by (\ref{struct}) with associated graph $G(\Sigma_{\Lambda})$.
System $\Sigma_{\Lambda}$ is structurally controllable if and only if
\begin{itemize}
  \item the graph $G(\Sigma_{\Lambda})$ is input-connected, and
  \item the state nodes of $G(\Sigma_{\Lambda})$ can be covered by a disjoint set of stems and cycles.
\end{itemize}
\end{thm}

\subsection{Linkings}\label{link}
Before moving on to characterise functional output controllability, we need to recall some additional graph notions.\\
Let us start with some graph $G=(V,E)$, with two possibly intersecting node subsets $V_1$ and $V_2$ of $V$. A path with initial node $v_1 \in V_1$ and terminal node $v_2 \in V_2$ is called a $(V_1-V_2)$-path. This path is said to be {\em direct} if $v_1$ (resp. $v_2$) is the only node of the path in $V_1$  (resp. $V_2$). A $(V_1-V_2)$-{\em linking} is a set of disjoint simple direct $(V_1-V_2)$-paths (with no nodes in common). For consistency, for $v \in V_1\cap V_2$, it is assumed that there is a $(V_1-V_2)$-path of length $0$ from $v$ to itself.
The size of a linking is the number of paths %\textcolor{red}{(of nonnegative length?)} 
it is composed of. 
A maximum $(V_1-V_2)$-linking is a linking of maximum size. 
Finding a maximum linking in the graph $G=(V,E)$ can be performed by using maximum flow techniques, and is then a problem with polynomial complexity \cite{MU:87,YAM:88}. 
This point will be treated in more detail in Section~\ref{Algo}.

When dealing with the graph $G(\Sigma_{\Lambda})$ of a structured system, if we choose $V_1=U$ and $V_2=Y$, we will speak about {\em input-output paths} and {\em input-output linkings}. Note that in this case, nodes of $U$ (resp. nodes of $Y)$ having no incoming edge (resp. no outgoing edge), any input-output path is necessarily direct. Input-output linkings have been a very convenient tool for the study of generic properties and control of structured systems, see for example \cite {CDP:91,VDW1:91,VDW2:91,DionCV:02}.

\subsection{Structural rank of the transfer matrix and functional output controllability}
When dealing with a structured system of type (\ref{struct}), the transfer matrix is $T_{\Lambda}(s)=C_{\Lambda}(sI_n-A_{\Lambda})^{-1}B_{\Lambda}$. 
The generic rank of $T_{\Lambda}(s)$ depends, in a very complex way, on the parameter vector $\Lambda$. 
However, it has been understood since the eighties \cite{CDP:91,OHTA:85,VDW1:91,VDW2:91} that the generic rank can be simply obtained from the graph $G(\Sigma_\Lambda)$.
\begin{thm}[Rank of transfer matrix~\cite {OHTA:85}] \label{RankT}
Let $\Sigma_\Lambda$  be the linear structured system defined by (\ref{struct}) with associated graph $G(\Sigma_\Lambda)$. 
The generic rank of the corresponding transfer matrix $T_\Lambda(s)$ is the size of a maximum input-output linking in $G(\Sigma_\Lambda)$.
\end{thm}
This result is not only remarkably simple, it is also rather intuitive. 
The maximum linking is indeed the maximum number of independent ways that the inputs may use to act on outputs. 
Notice that this result can also be seen as a generalisation of the characterisation of the rank of a structured matrix by means of the size of a maximum matching in the associated bipartite graph \cite{MU:87}.
\begin{exmp}\label{Example1}
Let us consider the structured system with 9 states, 2 inputs and 2 outputs defined by the following matrices:
%\begin{eqnarray}
$$A_\Lambda= \left(
\begin{array}{ccccccccc}
0 & 0 &0 &0 &0 &0 &0 &0 &\lambda_1\\ \lambda_2 & 0 &0 &0 &0 &0 &0 &0 &0 \\ 0 & \lambda_3 &0 &\lambda_4 &0 &0 &0 &0 &0 
\\ 0 & 0 &0 &0 &0 &0 &0 &0 &0 
\\ 0 & \lambda_5 &0 &\lambda_6 &0 &0 &0 &0 &0 
\\ \lambda_7 & 0 &0 &0 &\lambda_8 &0 &0 &0 &0 
\\ 0 & 0 &0 &0 &\lambda_{9} &0 &0 &0 &0
\\ 0 & 0 &0 &0 &\lambda_{10} &\lambda_{11} &\lambda_{12} &0 &0
\\ 0 & 0 &0 &0 &\lambda_{13} &\lambda_{14} &0 &0 &0
\end{array}
\right),$$

$$
B_\Lambda= \left(
\begin{array}{cc}
0 & 0 \\ 0 &0 \\ 0 &0 \\ \lambda_{15} &0 \\ 0 & 0 \\  0 &\lambda_{16} \\ \lambda_{17} &0 \\ 0 &0 \\ 0 &\lambda_{18} 
\end{array}
\right), $$
$$%&&\hspace{-0.50cm}
C_\Lambda = \left(
\begin{array}{ccccccccc}
0 & 0 &0 &0 &0 &0 &0 &\lambda_{19} &0 \\ 0 & 0 &0 &0 &0 &0 &0 & \lambda_{20} &\lambda_{21}
\end{array}
\right).
$$%\end{eqnarray}
The corresponding graph is given in Figure~\ref{Fig1}.
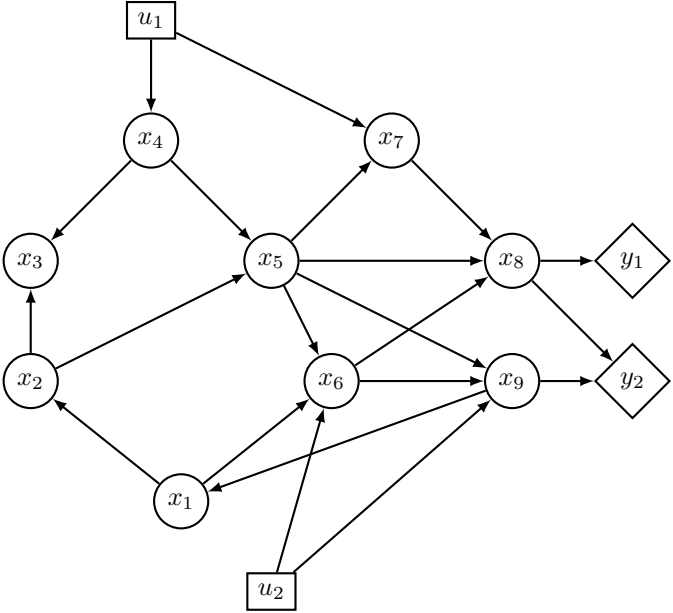
\begin{figure}%[htb]
 \centering
 \begin{tikzpicture}[scale=0.8]
\begin{scope}[every node/.style={circle,thick,draw}]
    \node (x1) at (2.5,2) {$x_1$};        
    \node (x2) at (0,4) {$x_2$};
    \node (x3) at (0,6) {$x_3$};
    \node (x4) at (2,8) {$x_4$};
    \node (x5) at (4,6) {$x_5$};
    \node (x6) at (5,4) {$x_6$};
    \node (x7) at (6,8) {$x_7$};
    \node (x8) at (8,6) {$x_8$};    
    \node (x9) at (8,4) {$x_9$};
    \node[rectangle,inner sep=4pt] (u1) at (2,10) {$u_1$};
    \node[rectangle, inner sep=4pt] (u2) at (4,0.5) {$u_2$};
    \node[diamond] (y1) at (10,6) {$y_1$};
    \node[diamond] (y2) at (10,4) {$y_2$};
\end{scope}

\begin{scope}[>=latex,
%>={Stealth[black]},
%              every node/.style={fill=white,circle},
              every edge/.style={draw,thick}]
%    \path [->] (x1) edge node {$5$} (B);
    \path [->] (x1) edge  (x2);
    \path [->] (x1) edge  (x6);
    \path [->] (x2) edge  (x3);
    \path [->] (x2) edge  (x5);
%    \path [->] (x2) edge  (x6);
    \path [->] (x4) edge  (x3);
    \path [->] (x4) edge  (x5);
    \path [->] (x5) edge  (x6);
    \path [->] (x5) edge  (x7);
    \path [->] (x5) edge  (x8);
    \path [->] (x5) edge  (x9);
    \path [->] (x6) edge  (x8);
    \path [->] (x6) edge  (x9);
    \path [->] (x7) edge  (x8); 
    \path [->] (x9) edge  (x1);
        \path [->] (x8) edge  (y1);
        \path [->] (x8) edge  (y2);
        \path [->] (x9) edge  (y2);
        \path [->] (u1) edge  (x4); 
        \path [->] (u1) edge  (x7);       
        \path [->] (u2) edge  (x6);
        \path [->] (u2) edge  (x9);
\end{scope}
\end{tikzpicture}

% \includegraphics[width=8cm]{Ex1.pdf}
% \vspace{-1.0cm}
 \caption{Graph $G(\Sigma_{\Lambda})$ of Example~\ref{Example1}. Input nodes $u_1$ and $u_2$ have rectangle shapes, output nodes $y_1$ and $y_2$ have diamond shapes.} \label{Fig1}
\end{figure}

From Theorem~\ref{RankT}, it can be seen that the generic rank of the transfer matrix  $T_\Lambda(s)$ is two. 
This follows from the fact that a maximal input-output linking in  $G(\Sigma_\Lambda)$ has size two. 
One can choose, for instance, the linking composed of the input-output paths $(u_1,x_4,x_5,x_8,y_1)$ and $(u_2,x_6,x_9,y_2)$.
\end{exmp}
Proposition~\ref{caract} and Theorem~\ref{RankT} can be combined to characterise the generic functional output controllability.
\begin{coro}[Structural characterisation of functional output controllability] \label{OFC}
Let $\Sigma_\Lambda$  be the linear structured system defined by
(\ref{struct}) with associated graph $G(\Sigma_\Lambda)$. 
The system $\Sigma_\Lambda$ is generically functional output controllable if and only if the size of a maximum  input-output linking in $G(\Sigma_\Lambda)$ is $p$, being the number of outputs.
\end{coro}
From Corollary~\ref{OFC}, the structured system of Example~\ref{Example1} is functional output controllable. 
Notice that by Theorem~\ref{lin} this Example~is also structurally controllable in the usual sense, since all the state nodes can be covered by the disjoint stems $(u_1,x_4,x_5,x_7,x_8)$ and $(u_2,x_6,x_9,x_1,x_2,x_3)$.

%Notice that the output controllability problem for structured systems, as defined by equation (\ref{Output-Cont}), was left as an open problem in \cite{MU:90}.
\subsection{On the difference between point-wise output controllability and functional output controllability}\label{sect:pointwise-vs-functional}
Consider a structured system whose graph is given in Figure~\ref{Fig2}, and whose matrices $(A_\Lambda,B_\Lambda)$ are 
\begin{eqnarray}
A_\Lambda= \left(
\begin{array}{cccc}
0 & 0 &0 & 0\\ \lambda_2 & 0& 0 & 0\\ 0 &\lambda_3 & 0 & 0\\ \lambda_4 & 0 & 0 & 0
\end{array}
\right), \quad B_\Lambda= \left(
\begin{array}{c}
\lambda_1\\ 0 \\ 0 \\0
\end{array}
%\right), \quad C_\Lambda = \left(
%\begin{array}{cc}
%\lambda_3 & 0\\
%0& \lambda_4
%\end{array}
\right).
\end{eqnarray}
The corresponding controllability matrix is
\begin{eqnarray}
K_\Lambda= \left(
\begin{array}{cccc}
\lambda_1 & 0 &0 & 0\\ 0 & \lambda_1 \lambda_2 & 0&0\\ 0 & 0 & \lambda_1 \lambda_2 \lambda_3 & 0\\ 0 & \lambda_1 \lambda_4 & 0 & 0
\end{array}
%\right), \quad B_\Lambda= \left(
%\begin{array}{c}
%\lambda_1\\ 0 \\ 0 \\0
%\end{array}
%\right), \quad C_\Lambda = \left(
%\begin{array}{cc}
%\lambda_3 & 0\\
%0& \lambda_4
%\end{array}
\right).
\end{eqnarray}
By Theorem~\ref{lin} this system is not structurally controllable, because the state nodes cannot be covered disjointly by stems and cycles. 
The controllability matrix $K_\Lambda$ has clearly generic rank three. 
It can be checked that we have target point-wise controllability for the target sets $\{x_1, x_2,x_3\}$ and $\{x_1, x_3, x_4\}$ (and their subsets, like for instance $\{ x_3 , x_4 \}$). 
This fact\footnote {This observation incidentally shows that the condition of Theorem~12 in \cite{BLA:10} is not necessary, because the nodes corresponding with the target set $\{x_3, x_4\}$ do not belong to a cactus in the graph.} follows from the generic independence of the corresponding rows in $K_\Lambda$.
Since the system has only one input, the maximal size of an input-output linking is one, and therefore the system is functionally controllable only for target sets composed of a unique variable. 
This shows that, strictly speaking, the values of output variables can be driven to any \textit{value} for target sets as  $\{x_1, x_2,x_3\}$ or $\{x_1, x_4\}$. 
But, clearly, with a unique input one cannot hope to follow any \textit{time profile} for two or more target variables simultaneously.\\
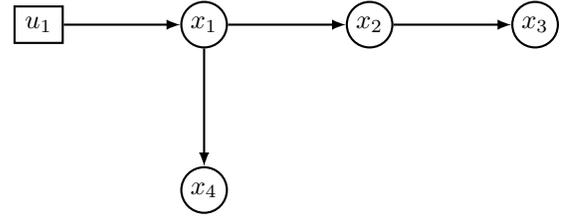
\begin{figure}%[htb]
 \centering
% \vspace{-0.5cm}
% \includegraphics[width=8cm]{Dif.pdf}
 
  \begin{tikzpicture}[scale=1.1]
\begin{scope}[every node/.style={circle,thick,draw},inner sep=2pt]
    \node (x1) at (0,0) {$x_1$};        
    \node (x2) at (2,0) {$x_2$};
    \node (x3) at (4,0) {$x_3$};
    \node (x4) at (0,-2) {$x_4$};
    \node[rectangle,inner sep=4pt] (u1) at (-2,0) {$u_1$};
\end{scope}
\begin{scope}[>=latex,     every edge/.style={draw, thick}]
    \path [->] (x1) edge  (x2);
    \path [->] (x1) edge  (x4);
    \path [->] (x2) edge  (x3);
     \path [->] (u1) edge  (x1);       
\end{scope}
\end{tikzpicture}
 \caption{Example illustrating the difference between point-wise controllability and functional controllability, as discussed in Section~\ref{sect:pointwise-vs-functional}.}\label{Fig2}
\end{figure}
\begin{rem}[On self-loops and functional controllability]An important feature of the graph characterisation of functional controllability in Corollary \ref{OFC} is that it only relies on input-output linkings, in which self-loops never appear. In other words, functional controllability is not dependent on the existence of self-loops, in contrast with point-wise controllability where they play a crucial role \cite{Liu:11,Cow:12}.
\end{rem}
\section{The Minimal Target Controllability Problem (MTCP) for complex networks} \label{MTCP}
%\subsection{Statement of the MTCP}
Let us now come to the problem as it appears in complex networks literature \cite{Liu:11,Ruths:14,OLS:15,CODI:13a,CODI:15,PEQ:16}. 
\subsection{Statement of the MTCP}
As in Problem (\ref{Prob}), we are given a dynamic system, representing the network, with target set of variables ${\mathcal T} \subset X$. 
We have now to choose a minimum number of steering nodes ${\mathcal S}$, which will define  an input set $U_{S}$, and therefore the $B_{\mathcal S}$ matrix, such that the system (\ref{lin2}) is functionally output controllable. 
When tackling this problem in the structured system framework, we have to find a minimum number of steering nodes such that the condition of Corollary~\ref{OFC} is satisfied. 
Since we are looking for a set of $p$ non intersecting  paths from inputs nodes to outputs nodes, it is clear that we must have at least $p$ steering nodes. 
On another hand, taking the target nodes as steering nodes gives a trivial minimal solution. 
The problem is in general more difficult, and also more interesting, because, due to physical considerations, not all nodes may be chosen as steering nodes \cite{OLS:15,CO:19}. 
The Minimal Target Controllability Problem (MTCP) can then be stated as follows. 
\begin{defn}[Minimal Target Controllability Problem (MTCP)]\label{def}
Given a structured system defined by matrices  $A_{\Lambda}$ and $C_{\mathcal T,\Lambda}$ related with a target node set ${\mathcal T}$, find a minimum number of steering nodes ${\mathcal S}$, taken from a given set of available nodes ${\mathcal A}=\{x_{i_1}, \dots,x_{i_k}\}$, such that the corresponding system of type (\ref{struct}) is generically functional target controllable.
\end{defn}
\subsection{Existence of a solution}
As mentioned previously, if there is no restriction on the possible steering nodes, {\em i.e.,} ${\mathcal A=\mathcal X}$, then the MTCP has a trivial solution. 
Instead, when ${\mathcal A \neq \mathcal X}$, the existence of such a solution is not always guaranteed. The conditions are detailed in the following proposition. 
\begin{prop}[Existence of a MTCP solution] \label{Exist}
Let $\Sigma_\Lambda$  be the linear structured system defined by
(\ref{struct}) with associated graph $G(\Sigma_\Lambda)$, where the input matrix $B_{\mathcal A,\Lambda}$ is related with the available node set ${\mathcal A}$, and the output matrix $C_{\mathcal T,\Lambda}$ is related with the target node set ${\mathcal T}$ of size $p$.
The Minimal Target Controllability Problem (MTCP) is solvable if and only if there exists an $({\mathcal A}$-${\mathcal T})$-linking of size $p$ in the graph $G(\Sigma_\Lambda)$. 
When a solution exists, the MTCP can be solved with a set of $p$ steering nodes. 
\end{prop}
%{\bf Proof}\\
\begin{proof}
$\Rightarrow$ If the MTCP has a solution, there exists a set of steering nodes ${\mathcal S} \subset{\mathcal A}$ such that the condition of Corollary~\ref{OFC} is satisfied, {\em i.e.,} there exists a $({U_\mathcal S}$-${Y_\mathcal T})$-linking of size $p$ in $G(\Sigma_\Lambda)$. 
This linking induces an $({\mathcal A}$-${\mathcal T})$-linking of size $p$, therefore the condition is necessary.\\
$\Leftarrow$ If there exists a $({\mathcal A}$-${\mathcal T})$-linking of size $p$ in  $G(\Sigma_\Lambda)$, the condition of Corollary~\ref{OFC} is satisfied by using ${\mathcal A}$ as the set of steering nodes. 
We can also choose as steering node set, the starting nodes of the $p$ paths of the previous linking, which provides a size $p$ solution to the MTCP.
\end{proof}
We emphasize that this proof of existence is constructive, in the sense that it describes a way to select the steering nodes that solve the MTCP (provided a solution exists). Since this solution is in general not unique, in the following section we shall have a closer look at the set of solutions.

\section{Importance of steering nodes}\label{Class1}
We have seen in the previous section that solutions to the MTCP are in general not unique. %In presence of multiple solutions, it is therefore interesting to classify nodes according to their importance 
Given a system $\Sigma_\Lambda$ and a set of available nodes ${\mathcal A}$ that satisfy the condition of Proposition~\ref{Exist}, in general there can exist several sets of steering nodes that provide a so-called {\em admissible} solution to the target controllability problem (minimal or not). It is therefore interesting to study what are the relations between such solution sets (for instance, whether there are nodes that necessarily belong to all solutions).
%A set of steering nodes which solves the  target controllability problem is called an {\em admissible steering node set}. 
Hence, the purpose of this section will be to classify the importance of nodes of ${\mathcal A}$ with respect to the target controllability problem. 
\begin{defn}[Classes of nodes for functional controllability]
\label{defclas} Let $x_i\in {\mathcal A}$.
\begin{itemize}
\item Node $x_i$ is said to be {\em essential} if $x_i \in {\mathcal D}$ for every admissible solution ${\mathcal D}$. 
\item Node $x_i$ is said to be {\em useless} if for any admissible solution ${\mathcal D}$ containing $x_i$, also ${\mathcal D}/\{x_i\}$ is an admissible solution. Otherwise, node $x_i$  is said to be {\em useful}. 
\end{itemize}
\end{defn}
Essential nodes are then particular useful nodes.
The classification of nodes in ${\mathcal A}$ will need the introduction of some new graph concepts. These concepts and results complete those of Subsection~\ref{link}.
\subsection{Separators}\label{sep}
Consider again a graph $G=(V,E)$, with two possibly intersecting node subsets $V_1$ and $V_2$ of $V$.
A $(V_1,V_2)$-{\em separator} is a set of nodes ${\mathcal S}$ such that every path from $V_1$ to $V_2$ covers a node in ${\mathcal S}$. 
The dependency on $V_1$ and $V_2$ is expressed by writing ${\mathcal S}(V_1,V_2)$.
The separator ${\mathcal S}(V_1,V_2)$ is said to be minimal if any proper subset of ${\mathcal S}(V_1,V_2)$ is not a separator between $V_1$ and $V_2$. 
It is a classical result of combinatorial optimisation that all minimal $(V_1,V_2)$-separators have the same size (cardinality) and that this size is equal to the size of a maximum linking between $V_1$ and $V_2$. \\
The minimal $(V_1,V_2)$-separator is generally not unique.
In this paper, a particular uniquely defined $(V_1,V_2)$-separator, called the {\em minimal left separator}, and denoted by ${\mathcal S}_*(V_1,V_2)$, will be used extensively.
%\margin{I find this precise definition very clarifying}
The set of minimal separators may be endowed with a partial ordering. 
Indeed, if $S$ and $T$ are minimal separators between $V_1$ and $V_2$, then $S$ is said to precede $T$, denoted ${S} \prec {T}$, when every direct path from $V_1$ to $V_2$ first passes through $S$ and next passes through $T$. 
The minimal left separator is the infimal minimal separator with respect to this order.
It is the $(V_1,V_2)$-separator of minimal size that is as close as possible to the node set $V_1$.
In simple words, the minimal left separator ${\mathcal S}_*(V_1,V_2)$ is the first smallest bottleneck that is met when we travel from $V_1$ to $V_2$. 
We collect in the following proposition some of the properties of ${\mathcal S}_*(V_1,V_2)$, which will be of interest for our purpose.
\begin{prop}[Minimal left separator] \label{Jacob}
Consider a graph $G=(V,E)$, with two node subsets $V_1$ and $V_2$ of $V$.
Assume that every node of $G$ is contained in a direct $(V_1-V_2)$-path. 
Then the following facts hold true.
\begin{enumerate}
\item The minimal left separator ${\mathcal S}_*(V_1,V_2)$ is uniquely defined and can be computed in polynomial time.
\item  If a node $v_i\in V_1$ does not belong to ${\mathcal S}_*(V_1,V_2)$, then there exists a maximal $(V_1-V_2)$-linking that covers $v_i$.
\item  If a node $v_i\in V_1$ does not belong to ${\mathcal S}_*(V_1,V_2)$, then there exists a maximal $(V_1-V_2)$-linking that does not cover $v_i$.
\end{enumerate}
\end{prop} 
%{\bf Proof}
\begin{proof}
\begin{enumerate}[leftmargin=*]
\item The result follows from \cite{VDW:00}, and computational details will be given in Section~\ref{Algo}.
\item  Assume that we are given a maximal $(V_1-V_2)$-linking that does not cover $v_i$. 
Consider  a $(V_1-V_2)$ direct path ${\mathcal P}$ with initial node $v_i$. Let $x_j$ be the first node at the intersection of ${\mathcal P}$ and of a path, denoted by ${\mathcal P}_k$, in the considered $(V_1-V_2)$-linking. 
A new $(V_1-V_2)$-linking of the same size can be constructed by replacing the part from a node of $V_1$ to  $x_j$ in the path ${\mathcal P}_k$, by the part from $v_i$ to $x_j$ in ${\mathcal P}$.
\item  The result follows from \cite{VDW:00}. 
A sketch of the proof is as follows.
If every maximal $(V_1-V_2)$-linking contains a path that covers node $v_i$, then $v_i$ must be contained in every (minimal) $(V_1-V_2)$-separator. In particular, $v_i$ is in the minimal $(V_1-V_2)$-separator that is closest to $V_1$, {\emph i.e.,} $v_i \in {\mathcal S}_*(V_1,V_2)$.
\end{enumerate}
\end{proof}
\subsection{Classification of available nodes}\label{av}
Proposition~\ref{Jacob} allows to give the complete classification of the nodes of the available set ${\mathcal A}$ with respect to the target set $\mathcal T$. This is simply obtained by choosing $V_1={\mathcal A}$ and $V_2={\mathcal T}$.
\begin{thm}[Classification of nodes] \label{Class2}
Consider the linear structured system defined by
(\ref{Alambda}) with associated graph $G(A_\Lambda)$. 
Let ${\mathcal A}$ be the set of available steering nodes and $\mathcal T$ be the target set. 
With respect to the Minimal Target Controllability Problem, the nodes of $\mathcal{A}$ can be classified as follows.
Considering node $x_i \in {\mathcal A}$, there holds the following.
\begin{enumerate}
\item Node $x_i$ is an essential steering node if and only if $x_i \in {\mathcal A}\cap {\mathcal S}_*({\mathcal A},{\mathcal T})$.
\item Node $x_i$ is a useless steering node if and only if there is no path from $x_i$ to a node in ${\mathcal T}$.
\end{enumerate}
\end{thm} 
%{\bf Proof}
\begin{proof}
\begin{enumerate}[leftmargin=*]
\item 
$\Leftarrow$ If $x_i$ is in ${\mathcal S}_*({\mathcal A},{\mathcal T})$, which is a minimal separator, by definition removing $x_i$ will decrease the size of a maximum input-output linking, therefore the MTCP has no solution, this implies that $x_i$ is essential. \\
$\Rightarrow$ If $x_i$ is not in ${\mathcal S}_*({\mathcal A},{\mathcal T})$, either it is not covered by a $({\mathcal A}-{\mathcal T})$-path and therefore does not belong to a maximal input-output linking, or, by point 2 of Proposition~\ref{Jacob}, there exists a maximal input-output linking which does not contain $x_i$.
Hence, in both cases $x_i$ is not essential.
\item $\Leftarrow$ If $x_i$ belongs to a solution, but is not covered by a $({\mathcal A}-{\mathcal T})$-path, there exists a maximal $({\mathcal A}-{\mathcal T})$-linking which does not cover $x_i$. 
Discarding $x_i$ will leave this  maximal $({\mathcal A}-{\mathcal T})$-linking unchanged, therefore the steering node set remains a solution without $x_i$, which is then useless.\\
$\Rightarrow$ If $x_i$ belongs to a solution and is covered by a $({\mathcal A}-{\mathcal T})$-path, from point 3 of Proposition~\ref{Jacob}, there exists a  maximal $({\mathcal A}-{\mathcal T})$-linking containing $x_i$. 
This linking provides a solution to the MTCP from which discarding $x_i$ would lead to a steering node set which is not a solution. Then $x_i$ is not useless.
\end{enumerate}
\end{proof}
%{\bf Example1, continued.}
\begin{exmp}\label{ex:sets-A-T}
Let us consider a network with the same dynamics as in Example~\ref{Example1}, with target nodes $\mathcal{T}=\{x_8,x_9\}$ and whose steering nodes have to be determined among the nodes of the available set ${\mathcal A}=\{x_1,x_2,x_3,x_4\}$. 
The corresponding graph is given in Figure~\ref{Fig3}.
\end{exmp}
\begin{figure}
 \centering

 \begin{tikzpicture}[scale=0.8]
\begin{scope}[every node/.style={circle,thick,draw}]
    \node (x1) at (2.2,2) {$x_1$};        
    \node (x2) at (0,4) {$x_2$};
    \node (x3) at (0,6) {$x_3$};
    \node (x4) at (2,8) {$x_4$};
    \node (x5) at (4,6) {$x_5$};
    \node (x6) at (5,4) {$x_6$};
    \node (x7) at (6,8) {$x_7$};
    \node (x8) at (8,6) {$x_8$};    
    \node (x9) at (8,4) {$x_9$};
%    \node (u1) at (2,10) {$u_1$};
%   \node (u2) at (4,1) {$u_2$};
%    \node (y1) at (10,6) {$y_1$};
 %   \node (y2) at (10,4) {$y_2$};
\end{scope}

\begin{scope}[>=latex,
%>={Stealth[black]},
%              every node/.style={fill=white,circle},
              every edge/.style={draw, thick}]
%    \path [->] (x1) edge node {$5$} (B);
    \path [->] (x1) edge  (x2);
    \path [->] (x1) edge  (x6);
    \path [->] (x2) edge  (x3);
    \path [->] (x2) edge  (x5);
%    \path [->] (x2) edge  (x6);
    \path [->] (x4) edge  (x3);
    \path [->] (x4) edge  (x5);
    \path [->] (x5) edge  (x6);
    \path [->] (x5) edge  (x7);
    \path [->] (x5) edge  (x8);
    \path [->] (x5) edge  (x9);
    \path [->] (x6) edge  (x8);
    \path [->] (x6) edge  (x9);
    \path [->] (x7) edge  (x8); 
    \path [->] (x9) edge  (x1);
%        \path [->] (x8) edge  (y1);
%        \path [->] (x9) edge  (y2);
%        \path [->] (u1) edge  (x4);       
%        \path [->] (u2) edge  (x6);
\end{scope}

\node[draw, rectangle, very thick, green, fit=(x1) (x2) (x2) (x4), label=above:{\color{green}$\mathcal{A}$}, inner sep=7pt] {}; %ellipse
\node[draw, diamond, very thick, blue, fit=(x8) (x9), label=above:{\color{blue}$\mathcal{T}$}] {};

\end{tikzpicture}
\caption{Steering node selection for target controllability of Example~\ref{ex:sets-A-T} with target set $\mathcal{T}$ highlighted by a blue diamond. Within set $\mathcal A$ (green rectangle), node $x_1$ is essential, node $x_3$ is useless, nodes $x_2$ and $x_4$ are useful.  \label{Fig3}}
\end{figure}
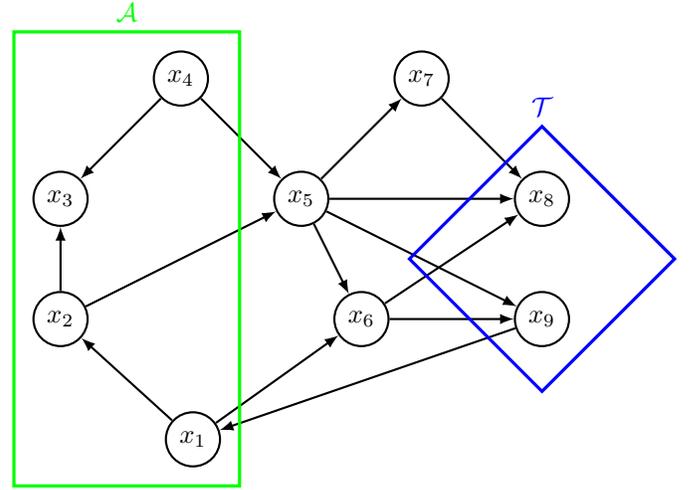

%\margin{PF: $\mathcal{T}$ and $Y$ seem to denote the same thing... should we write ${\mathcal S}_*({\mathcal A},{\mathcal T})$?}
From Proposition~\ref{Exist}, it is clear that the MTCP has a solution. 
For example, we have the size two linking $\{(x_1,x_6,x_9), (x_2,x_5,x_8)\}$. 
Therefore, the nodes $x_1$ and $x_2$ can be chosen as steering nodes for a minimal solution.\\
Let us now examine the importance of the different nodes of ${\mathcal A}$ with respect to the target controllability problem. 
The minimum input separator ${\mathcal S}_*({\mathcal A},{\mathcal T})$ can shown to be equal to $\{x_1,x_5\}$. 
Then from Theorem~\ref{Class2}, node $x_1$ is essential for the target controllability problem, {\em i.e.,} it belongs to all the solutions of the problem. On the contrary, since node $x_3$ belongs to no $({\mathcal A},{\mathcal T})$-path, it is useless for the target controllability problem. Finally, nodes $x_2$ and $x_4$ are simply useful steering nodes.
\begin{rem}[Extension to functional output controllability]
Both the characterisation of a solution in Proposition~\ref{Exist} and the classification of Theorem~\ref{Class2} were stated with respect to a particular target set $\mathcal{T}$. 
%\margin{\textcolor{red}{I don't really understand the meaning of this remark}}
However, in the proofs, no use was made of the particular form of the corresponding matrix $C_{\mathcal T,\Lambda}$. 
Therefore, all these results remain valid for a general output matrix $C_{\Lambda}$ and could be stated in terms of functional  output controllability instead of functional target controllability.
\end{rem}

\section{Algorithmic and complexity aspects}\label{Algo}
%\subsection{Algorithm}
%Finding a maximal input-output linking in a graph is in fact equivalent to solving a maximum flow problem in a modified graph \cite{HOV:96,YAM:88}.
% The equality between the size of maximum input-output linking and the size of a minimal separator is just a consequence of the Theorem~of Menger, which in turn is closely related to the famous Max-flow Min-cut Theorem~\cite{MU:87}. \\
%\margin{I removed the blue color, but let us not forget that this section is relatively new}

In this section, we seek to exploit the structural results of Section~\ref{Class1} to propose efficient algorithms to solve the MTCP and to classify the available nodes. 
From Theorem~\ref{Class2} item 2, we know that the {\em useless nodes} in $\mathcal{A}$ are the nodes which are not the initial node of a $(\mathcal{A}-\mathcal{T})$-path. This property can be checked with a depth first algorithm whose complexity is linear in the number of edges of the graph. At the same time, from Theorem~\ref{Class2} item 1, we know that the determination of the essential nodes requires the computation of the minimal left separator $(\mathcal{A}-\mathcal{T})$. We will prove that this set can be obtained from the application of the Ford-Fulkerson algorithm on an auxiliary graph. Preliminarily to showing this fact, we observe that since ${\mathcal S}_*({\mathcal A},{\mathcal T})$ is only related to direct $(\mathcal{A}-\mathcal{T})$-paths, {\em w.l.o.g.} we can delete all the incoming edges in $\mathcal{A}$ and all the outgoing edges from $\mathcal{T}$.
\subsection{Separators and cuts for the essential nodes}\label{sepcut}
Let us briefly recall some basics of flow theory \cite{FF:62}. 
Consider first a graph $G(V,E)$ with two distinguished vertices, a source $s$ with no incoming edge and a sink $t$ with no outgoing edge. 
A flow is a real number $f(e)$ associated with each edge $e$ of the graph, which satisfies the balance equation ({\em i.e.,} for each node, except for the source and the sink, the incoming flow equals the outgoing flow). 
A non-negative integer capacity $c(e)$ is associated with each edge and a flow is said to be feasible if for each edge $e$ of the graph, $0\leq f(e) \leq c(e)$. 
Notice that a flow $0$ on each edge is feasible. 
For a graph with a feasible flow, an augmenting path is defined as an undirected path from $s$ to $t$ ({\em i.e.,} a path containing forward and backward edges) which is such that for each forward edge $e$, we have $f(e)<c(e)$ and for each backward edge $e'$, we have $f(e')>0$. 
The existence of an augmenting path gives the possibility to obtain a new feasible flow which is greater than the previous one. A source set is a set of vertices ${\bar V}$ such that $s\in {\bar V}$ and $t\notin {\bar V}$. The {\em cut} associated with the source set ${\bar V}$ is the set of edges $(v,v') \in E$ such that the initial node $v$ is in ${\bar V}$ and the terminal node $v'$ is not in ${\bar V}$.
The capacity $c({\bar V})$ of the cut is defined as the sum of the capacities of the edges it is composed of. 
The famous Max-flow Min-cut Theorem~\cite{FF:62,MU:87} states that the maximal flow from $s$ to $t$ in the graph $G(V,E)$ is equal to the minimal capacity of a cut.

We are now ready to bear on these notions and algorithms from flow theory to state and prove our final result. Our key instrument will be the following definition of auxiliary graph.
\begin{defn}[Auxiliary graph]
Consider a system $A_\Lambda$ with available set $\mathcal{A}$ and target set $\mathcal{T}$ and its graph $G(A_\Lambda)$. 
We define an associated auxiliary graph $\subscr{G}{aux}(A_{\Lambda})$ as follows:
%\margin{PF: Shall we construct and display an example, eg, for the graph of Example1?}
\begin{itemize}
\item Split each state node $x_i$ of  $G(A_\Lambda)$ into two nodes $x_i^-$ and $x_i^+$, and add an edge $(x_i^-,x_i^+)$.
\item Transform each edge of  $G(A_\Lambda)$ of the form $(x_i,x_j)$ into an associated edge of $\subscr{G}{aux}(A_{\Lambda})$ of the form $(x_i^+,x_j^-)$.
\item Create in $\subscr{G}{aux}(A_{\Lambda})$ a dummy source node $s$ and a dummy sink node $t$, add an edge from $s$ to all the available nodes $\{x_{i1}^-, x_{i2}^-,\ldots,x_{im}^-\}$, and an edge from all target nodes $\{x_{j1}^+, x_{j2}^+,\ldots,x_{jp}^+\}$ to the sink node $t$.
\item Give to all the edges $(x_i^-,x_i^+)$, for $i=1,\dots,n$, a capacity one, and to all other edges of $\subscr{G}{aux}(A_{\Lambda})$ an infinite capacity.
\end{itemize}
\end{defn}
%It is easy to show, see for Example~\cite{YAM:88}, that finding a maximum flow in $\subscr{G}{aux}(A_{\Lambda})$ is equivalent to finding a maximum input-output linking in $G(A_\Lambda)$. 
%Moreover, when applying the well-known Ford-Fulkerson algorithm \cite{FF:62} to find the maximum flow in $\subscr{G}{aux}(A_{\Lambda})$, the first minimum cut which will be met in the augmenting path procedure corresponds exactly to the minimal input separator ${\mathcal S}_*$ in $G(A_\Lambda)$. 

This definition allows us to state the following result.

\begin{prop}[Max flow and Min cut in MTCP]\label{Sep-Cut}
Consider a structured system $A_\Lambda$ with available set $\mathcal{A}$ and target set $\mathcal{T}$ with its graph $G(A_\Lambda)$, and the associated auxiliary graph  $\subscr{G}{aux}(A_{\Lambda})$. The following two facts hold.
\begin{itemize}
\item The size of a maximal $({\mathcal A},{\mathcal T})$-linking in $G(A_\Lambda)$ is the value $F$ of a maximal flow on $\subscr{G}{aux}(A_{\Lambda})$. 
\item The minimal left separator ${\mathcal S}_*({\mathcal A},{\mathcal T})$ of $G(A_\Lambda)$ is in one-to-one correspondence with the minimal cut in $\subscr{G}{aux}(A_{\Lambda})$ that is produced by the Ford-Fulkerson algorithm.
\end{itemize}
\end{prop}
\begin{proof}%(Sketch)
From the construction of the  auxiliary graph $\subscr{G}{aux}(A_{\Lambda})$, every separator of $G(A_\Lambda)$ induces a cut in $\subscr{G}{aux}(A_{\Lambda})$ whose capacity is the size of the separator. 
Moreover, a cut in $\subscr{G}{aux}(A_{\Lambda})$ has a finite capacity only if it corresponds with a separator of $G(A_\Lambda)$, otherwise the cut would contain an edge with an infinite capacity. Therefore, there is a one-to-one correspondence between separators in $G(A_\Lambda)$ and finite cuts in $\subscr{G}{aux}(A_{\Lambda})$, the size of the separator being equal to the capacity of the cut. As a consequence, any minimum separator of $G(A_\Lambda)$ is in a one-to-one correspondence with a minimum cut of $\subscr{G}{aux}(A_{\Lambda})$. 
The first item of Proposition~\ref{Sep-Cut} then follows from the Max-Flow Min-Cut Theorem, see also \cite{YAM:88,HOV:96}. 

We now prove the second item. Starting from an initial null flow on each edge, the Ford-Fulkerson algorithm~\cite{FF:62}
% (for which a detailed pseudo-code is given in Algorithm~\ref{algo:ff} for completeness)
  is iteratively composed of two phases. 
In the first phase, a labelling procedure, starting in $s$, looks for an augmenting path. 
If no augmenting path is found, the algorithm stops and the actual flow is indeed maximum. 
If an augmenting path is found, the flow is increased along this path. 
The algorithm stops when the flow is maximum and when the set of labelled nodes is the source set associated with the minimal cut, which is the closest to the source \cite{PIC:80}. 
From the previous observation on the correspondence between minimum cuts in $\subscr{G}{aux}(A_{\Lambda})$ and the minimum $({\mathcal A},{\mathcal T})$-separators in $\subscr{G}{aux}(A_{\Lambda})$, the second item of Proposition~\ref{Sep-Cut} follows.
\end{proof}
%When dealing with our problem, {\em i.e.} finding a minimal set of steering nodes from a set of available nodes for controlling a set of target nodes, and classifying the available nodes with respect to this objective, the auxiliary graph simplifies slightly. 
%In this case, we only need to connect the source $s$ to nodes $x_i^-$, where $x_i$ is an available node, and $x_j^+$ to the sink $t$, for every target node $x_j$. 
The auxiliary graph associated to our Example~of Figure~\ref{Fig3} is given in Figure~\ref{Fig4}. 
In Figure~\ref{Fig4}, it appears that starting from a zero flow, a first augmenting path $(s,x_2^-,x_2^+,x_5^-,x_5^+,x_8^-,x_8^+,t)$ composed only of forward edges allows to convey a unit flow from $s$ to $t$. 
Similarly, a second path $(s,x_1^-,x_1^+,x_6^-,x_6^+,x_9^-,x_9^+,t)$ allows to convey a supplementary unit flow from $s$ to $t$. 
When this flow is installed on the graph, the labelling procedure allows to reach the node set $L=\{s,x_4^-,x_4^+,x_5^-,x_3^-,x_3^+,x_2^-,x_1^-,x_2^+\}$. 
This source set is associated to a minimal cut and the corresponding flow is a maximum one. This induces that
\begin{itemize}
\item the minimum number of steering nodes for target controllability is two, and that $\{x_1,x_2\}$ is a possible solution,
\item the set ${\mathcal S}_*({\mathcal A},{\mathcal T})$ is equal to $\{x_1,x_5\}$ which implies that $x_1$ is an essential node.
\end{itemize}
%\todo{please use the Example~(in the new figure) to illustrate the algorithm}
\subsection{Complexity of finding the essential nodes}
%\begin{rem}[Computational complexity]
The complexity of the Ford-Fulkerson algorithm with integer capacities is of order $O(N_e \cdot f_M)$, where $N_e$ is the number of edges of the graph, and $f_M$ is the value of the maximum flow. 
The number of edges in $\subscr{G}{aux}(A_{\Lambda})$ being bounded by $(2n+2)^2$, and the flow being bounded by $p$, we finally get a complexity of order $O(n^2p)$. 
There are certainly better performing maximum flow algorithms, but it is important to note that the Ford-Fulkerson algorithm also provides the minimal left separator  ${\mathcal S}_*({\mathcal A},{\mathcal T})$ \cite{PIC:80}.

\begin{figure*}[htb]
 \centering

 \begin{tikzpicture}[scale=1.25]
\begin{scope}[every node/.style={circle,thick,draw,inner sep=1pt}]
    % "minus" nodes
    \node (x1m) at (3,2.5) {$x_1^-$};        % this node is the only one that has "moved" in order not to mess the edges up
    \node (x2m) at (0,4) {$x_2^-$};
    \node (x3m) at (0,6) {$x_3^-$};
    \node (x4m) at (1.5,8) {$x_4^-$};
    \node (x5m) at (3.7,6) {$x_5^-$};
    \node (x6m) at (4.5,4) {$x_6^-$};
    \node (x7m) at (6,8) {$x_7^-$};
    \node (x8m) at (8,6) {$x_8^-$};    
    \node (x9m) at (8,4) {$x_9^-$};
%    \node (u1m) at (2,10) {$u_1^-$};
%   \node (u2m) at (4,1) {$u_2^-$};
%    \node (y1m) at (11,6) {$y_1^-$};
%    \node (y2m) at (11,4) {$y_2^-$};
    
    % "plus" nodes
    \node (x1p) at (5,2.5) {$x_1^+$};        
    \node (x2p) at (1.5,4) {$x_2^+$};
    \node (x3p) at (1.5,6) {$x_3^+$};
    \node (x4p) at (3,8) {$x_4^+$};
    \node (x5p) at (5.5,6) {$x_5^+$};
    \node (x6p) at (6,4) {$x_6^+$};
    \node (x7p) at (7.5,8) {$x_7^+$};
    \node (x8p) at (9.5,6) {$x_8^+$};    
    \node (x9p) at (9.5,4) {$x_9^+$};
%    \node (u1p) at (3.5,10) {$u_1^+$};
%   \node (u2p) at (5.5,1) {$u_2^+$};
%    \node (y1p) at (12.5,6) {$y_1^+$};
%    \node (y2p) at (12.5,4) {$y_2^+$};  
\end{scope}
    % special nodes
\node[inner sep=5pt,thick,draw,rectangle] (s) at (-1.5,5) {$s$}; 
\node[inner sep=3pt,thick,draw,diamond] (t) at (11,5) {$t$};

% infinite capacity edges
\begin{scope}[>=latex,
%>={Stealth[black]},
%              every node/.style={fill=white,circle},
              every edge/.style={draw,thick}]
%    \path [->] (x1) edge node {$5$} (B);
    \path [->] (x1p) edge  (x2m);
    \path [->] (x1p) edge  (x6m);
    \path [->] (x2p) edge  (x3m);
    \path [->] (x2p) edge  (x5m);
%    \path [->] (x2p) edge  (x6m);
    \path [->] (x4p) edge  (x3m);
    \path [->] (x4p) edge  (x5m);
    \path [->] (x5p) edge  (x6m);
    \path [->] (x5p) edge  (x7m);
    \path [->] (x5p) edge  (x8m);
    \path [->] (x5p) edge  (x9m);
    \path [->] (x6p) edge  (x8m);
    \path [->] (x6p) edge  (x9m);
    \path [->] (x7p) edge  (x8m); 
    \path [->,bend left] (x9p) edge  (x1m);
        \path [->] (x8p) edge  (t);
        \path [->] (x9p) edge  (t);
%        \path [->] (x9p) edge  (y2m);
%        \path [->] (u1p) edge  (x4m);       
%        \path [->] (u2p) edge  (x6m);
        
%        \path [->] (s) edge  (x1m);
   		\path [->] (s) edge  (x2m);
        \path [->] (s) edge  (x3m);
        \path [->,bend left] (s) edge  (x4m);
        \path [->,bend right] (s) edge  (x1m);
%        \path [->] (y1p) edge  (t);
%        \path [->] (y2p) edge  (t);                
\end{scope}

% capacity 1 edges
\begin{scope}[>=latex,
%>={Stealth[black]},
%              every node/.style={fill=white,circle},
              every edge/.style={draw,very thick,dashed}]
%    \path [->] (x1) edge node {$5$} (B);
    \path [->] (x1m) edge  (x1p);
\path [->] (x2m) edge  (x2p);
\path [->] (x3m) edge  (x3p);
\path [->] (x4m) edge  (x4p);
\path [->] (x5m) edge  (x5p);
\path [->] (x6m) edge  (x6p);
\path [->] (x7m) edge  (x7p);
\path [->] (x8m) edge  (x8p);
\path [->] (x9m) edge  (x9p);
%\path [->] (u1m) edge  (u1p);
%\path [->] (u2m) edge  (u2p);
%\path [->] (y1m) edge  (y1p);
%\path [->] (y2m) edge  (y2p);
\end{scope}

\node[draw, very thick, rounded corners, densely dashed, red, fit=(s) (x4m) (x4p) (x5m) (x3m) (x3p) (x2m) (x2p) (x1m), label=above:{min cut}] {}; %ellipse, rounded corners

\end{tikzpicture}
\caption{Auxiliary graph associated to the graph of Figure~\ref{Fig3}. \label{Fig4} Dummy source node $s$ has rectangle shape, dummy sink node $t$ has diamond shape. Edges with infinite capacity are drawn as solid lines, edges with capacity one are drawn as dashed lines. A red dashed box encloses the source set associated to the minimal cut that is identified by the Ford-Fulkerson algorithm.}
\end{figure*}
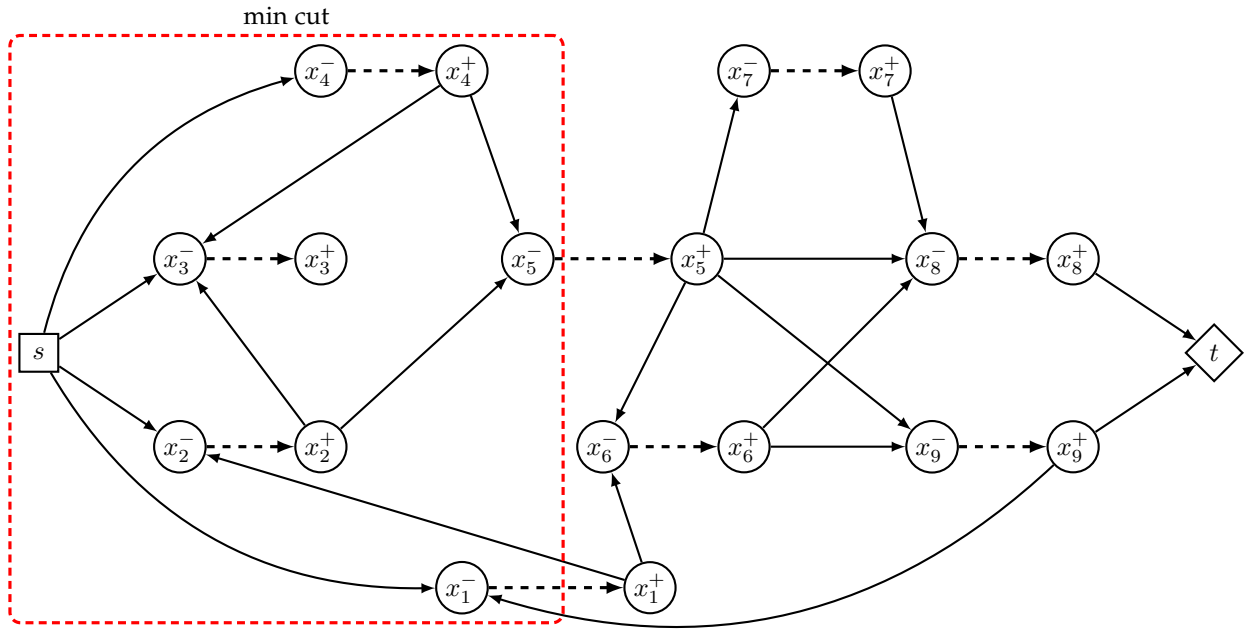

\section{Conclusion}\label{Conc}
In this paper we have introduced the notion of functional target controllability. 
This notion is a relevant alternative to the classical point-wise target controllability. 
We think that this notion is justified by the fact that the importance of target variables needs a refined type of controllability. 
It happens that this new approach induces a simpler characterisation of target controllability in graph terms for structured systems than the classical point of view. 
This opens the possibility to revisit some problems as the robustness of target controllability against edge deletions~\cite{Fardad:14}, or the opposite problem of finding necessary edge additions to reach  target controllability \cite{CHEN:18}. 
%\appendix
%\section{References}
%==============================BIBLIOGRAPHY================================
\bibliographystyle{IEEEtran}
\bibliography{Target}

\end{document}